\documentclass[smallextended,11pt,psfig,a4]{amsart}

\usepackage{mleftright} 

\usepackage{adjustbox}

\usepackage{amsmath}
\usepackage{amsfonts}
\usepackage{amssymb}
\usepackage{cancel}
\usepackage{bm}
\usepackage{braket}

\usepackage{amsthm}
\usepackage{bbm}
\usepackage{MnSymbol}	
\usepackage[dvipsnames]{xcolor}
\usepackage{mathrsfs}

\usepackage[foot]{amsaddr}

\usepackage{geometry}
\usepackage[pdftex]{hyperref}
\usepackage{amssymb}
\usepackage[normalem]{ulem}
\usepackage{xcolor,soul}

\DeclareMathOperator{\Ind}{Ind}
\DeclareMathOperator{\ger}{span}

\DeclareMathOperator{\Ran}{Ran}

\DeclareMathOperator{\vect}{vec}

\newtheorem{theorem}{Theorem}[section]
\newtheorem{lemma}[theorem]{Lemma}
\newtheorem{pro}[theorem]{Proposition}
\newtheorem*{conj*}{Conjecture}

\newtheorem{remark}[theorem]{Remark}

\theoremstyle{definition}
\newtheorem{definition}[theorem]{Definition}
\newtheorem{example}[theorem]{Example}

\theoremstyle{remark}

\numberwithin{equation}{section}

\newcommand{\bracket}[2]{\langle{#1}|{#2}\rangle}
\newcommand{\ceil}[1]{\lceil {#1}\rceil}

\newcommand{\opn}[1]{\operatorname{#1}}

\newcommand{\mbb}[1]{\mathbb{#1}}

\newcommand{\bs}[1]{\boldsymbol{#1}}

\DeclareMathOperator{\diag}{diag}

\def\eps{\varepsilon}
\def\>{\rangle}
\def\<{\langle}

\def\Tr{\opn{Tr}}

\def\0{\bs{0}}
\def\1{\mathbbm{1}}

\def\C{\mbb{C}}

\def\HH{\mathscr{H}}

  \def\XXint#1#2#3{{\setbox0=\hbox{$#1{#2#3}{\int}$}
      \vcenter{\hbox{$#2#3$}}\kern-.47\wd0}}

\usepackage{tikz}

\tikzstyle{vertex}=[circle, draw, inner sep=0pt, minimum size=4pt]

\usepackage{xcolor,soul}

\usetikzlibrary{arrows,shapes,decorations,automata,backgrounds,petri,positioning}

\tikzset{main node/.style={circle,draw,minimum size=1cm,inner sep=0pt},
}

\begin{document}
\pagestyle{plain}

\def\beq{\begin{equation}}
\def\eeq{\end{equation}}
\def\eps{\epsilon}
\def\laa{\langle}
\def\raa{\rangle}
\def\qed{\begin{flushright} $\square$ \end{flushright}}
\def\qee{\begin{flushright} $\Diamond$ \end{flushright}}
\def\ov{\overline}
\def\bma{\begin{bmatrix}}
\def\ema{\end{bmatrix}}

\def\ora{\overrightarrow}

\def\bma{\begin{bmatrix}}
\def\ema{\end{bmatrix}}
\def\bex{\begin{example}}
\def\eex{\end{example}}
\def\beq{\begin{equation}}
\def\eeq{\end{equation}}
\def\eps{\epsilon}
\def\laa{\langle}
\def\raa{\rangle}
\def\qed{\begin{flushright} $\square$ \end{flushright}}
\def\qee{\begin{flushright} $\Diamond$ \end{flushright}}
\def\ov{\overline}

\author[cfelipe]{C. F. Lardizabal}
\address{Instituto de Matem\'atica e Estat\'istica, Universidade Federal do Rio Grande do Sul, Porto Alegre, RS  91509-900 Brazil.}
\email{cfelipe@mat.ufrgs.br (Corresponding author)}

\author[lv]{L. F. L. Pereira}
\address{Instituto de Matem\'atica e Estat\'istica, Universidade Federal do Rio Grande do Sul, Porto Alegre, RS  91509-900 Brazil.}

\date{\today}

\title{Hitting time expressions for quantum channels: beyond the irreducible case and applications to unitary walks}

\begin{abstract} 
In this work we make use of generalized inverses associated with quantum channels acting on finite-dimensional Hilbert spaces, so that one may calculate the mean hitting time for a particle to reach a chosen goal subspace. The questions studied in this work are motivated by recent results on quantum dynamics on graphs, most particularly quantum Markov chains. We focus on describing how generalized inverses and hitting times can be obtained, with the main novelties of this work with respect to previous ones being that a) we are able to weaken the notion of irreducibility, so that reducible examples can be considered as well, and b) one may consider arbitrary arrival subspaces for general positive, trace preserving maps. Natural examples of reducible maps are given by unitary quantum walks. We also take the opportunity to explain how a more specific inverse, namely the group inverse, appears in our context, in connection with matrix algebraic constructions which may be of independent interest.
\end{abstract}

\maketitle



\section{Introduction}\label{sec1}

In quantum information theory, there is a natural interest in the problem of finding
quantum versions of classical probabilistic notions. For instance, just like stochastic matrices acting on probability vectors are the objects that describe the statistical evolution of Markov chains associated with classical particles on some graph \cite{bremaud,kemeny,norris}, one may say that their quantum counterparts are the so-called quantum channels, that is, maps of the form
$$T:M_n\to M_n,\;\;\;T(\rho)=\sum_{i} V_i\rho V_i^*,$$
where $M_n$ denotes the order $n$ matrices, and $V_i\in M_n$ satisfy $\sum_i V_i^*V_i=I$, which is equivalent to say that $T$ is  trace preserving \cite{bhatia,wolf}. We usually require that such maps act on density matrices, that is, on positive semidefinite matrices of trace 1. We note that if $U\in M_n$ is a unitary matrix then the conjugation map $\rho\mapsto U\rho U^*$ is an important particular case of this setting.

\medskip

In this work we are interested in statistics associated with quantum particles acting on some finite-dimensional Hilbert space. Namely, we will study the mean time of first visit of a quantum particle to a goal subspace under the iterates of a quantum channel. Notions of quantum hitting times are well-known in the context of unitary quantum walks, most notably Szegedy's walk \cite{portugal,szegedy}, and considerable detail on their statistics are known \cite{salvador}. On the other hand, hitting times associated with quantum channels seem to be less understood and, just like the unitary case, more than one definition is possible. One way of defining such a notion is in terms of a so-called monitoring procedure, for which we perform projection operators onto the subspaces of interest, and their complements, to verify whether a particle has been found. This approach has been explored in several works in recent years, see for instance \cite{ambainis,glv,werner,kuklinski,gawron}, and will be a basic element of our setting.

\section{Motivation: hitting times of quantum dynamics and their generalized inverses}\label{intro2}

The question we will address in this work is motivated by the following setting, coming from the classical theory of Markov chains: given a graph and transition probabilities between its vertices, what is the mean time for a walker to reach vertex $j$ for the first time, given that it has started at vertex $i$? Formally, the mean hitting time is given by
$$m_{ij}=E_i(T_j)=\sum_{t=1}^\infty tP_i(T_j=t),$$
where $T_j$ is the random variable given by the time of first visit to vertex $j$, and $P_i(T_j=t)$ is the probability that $T_j=t$, given that the walk begins at position $i$. 

\medskip

From the theory of Markov chains we know that, alternatively, the mean hitting time can be calculated without resorting to its definition directly. A well-known method is via the {\bf fundamental matrix} associated with a finite, irreducible Markov chain with stochastic matrix $P$,
$$Z=(I-P+\Omega)^{-1},$$
where $\Omega$ is the matrix whose columns are given by $w=(w_i)$, the unique fixed probability vector for $P$. Assuming $P$ is column stochastic, the following equation is valid:
\beq\label{clmhtf} E_i(T_j)=\frac{Z_{jj}-Z_{ji}}{w_j}.\eeq
This is the {\bf mean hitting time formula (MHTF)}, and is one of several expressions relating $Z$ with statistical quantities of the walk \cite{aldous,bremaud,kemeny}. 

\medskip

In a similar vein,  we have the following result. If $A$ is a matrix, denote by $A_d$ the diagonal matrix whose diagonal entries are the ones from $A$. Let $e=[1\; 1 \; \cdots \; 1]^T$ and let $E=ee^T$, that is, the matrix whose entries are all equal to 1.

\begin{theorem}\label{ht1}\cite{hunter} Let $M=(m_{ij})$ be the mean time of first visit matrix for an irreducible Markov chain with column stochastic transition matrix $P$ and stationary probability vector $w$. Then $G=(I-P+ue^T)^{-1}+fe^T$ is a generalized inverse of $I-P$, for any choice of vectors $f$ and $u$ with $\langle u|w\rangle \neq 0$. Moreover, if $\Omega=w e^T$ and $D=(\Omega_d)^{-1}$, then
\beq\label{htf1} M=D[I-G+G_dE].\eeq
\end{theorem}

The above statement is adapted from the article due to Hunter \cite{hunter} but, up to our knowledge, such formula has appeared for the first time in the monograph due to Kemeny and Snell \cite{kemeny}, where matrix $G$ is actually the fundamental matrix $Z$, also see the work due to Meyer \cite{meyer-artigo}. We will call expression (\ref{htf1}) the {\bf Kemeny-Snell-Meyer-Hunter (KSMH) formula}.

\medskip

It can be shown that the fundamental matriz $Z$ is also a generalized inverse of $I-P$ and, moreover, that the MHTF (\ref{clmhtf}) is a particular case of Theorem \ref{ht1}. 

\medskip

Motivated by these classical results, the following is a natural question:

\medskip

{\bf Question:} regarding dynamics described by a quantum channel acting on a finite-dimensional Hilbert space, is it possible to obtain mean hitting time expressions in terms of generalized inverses associated with the evolution?

\medskip

In the following section we will give a positive answer to this question, first reviewing what is known and then stating the new result.

\medskip

\section{Overview of results and structure of the work}\label{intro3}

In the context of quantum information, regarding the calculation of mean hitting times of quantum channels in terms of generalized inverses, the following results were obtained so far:

\medskip

\begin{enumerate}
\item A version of the MHTF (\ref{clmhtf}) for irreducible Open Quantum Walks (OQWs) on finite graphs was proved in \cite{c2017}. OQWs are a particular class of quantum channels for which the underlying Hilbert space is identified with the nodes of some chosen graph, together with the internal degrees of freedom of the quantum particle at each node \cite{attal}. 

\medskip

\item A version of the MHTF (\ref{clmhtf}) was proved for irreducible positive maps on finite-dimensional Hilbert spaces \cite{cv}, thus including the result from \cite{c2017} as a particular case.

\medskip

\item A version of the KSMH formula (\ref{htf1}) for the setting of irreducible quantum Markov chains (QMCs) on finite graphs was proved in \cite{oqwmhtf2}. QMCs were defined by S. Gudder \cite{gudder}, and include OQWs as a particular class of channels.
\end{enumerate}

QMCs are briefly revised in this work, as their structure will be of assistance for the proof of our new result.

\medskip

We remark that in all of the works above, we have the assumption that the operators are irreducible, which can be seen as a kind of ``connectivity'' of the dynamics. As we will be considering particles with internal degrees of freedom, one should work with a careful, precise definition regarding the accessibility of subspaces. This will be properly revised in this work.

\medskip

With the above results in mind, we may state the following:

\medskip

{\bf Problem:} regarding quantum dynamics described by an irreducible quantum channel acting on a finite-dimensional Hilbert space, we know that it is possible to study mean hitting times in terms of generalized inverses. On the other hand, is it possible to obtain hitting time formulae without the assumption of irreducibility?

\medskip

The above problem goes beyond the mathematical discussion about dropping a technical assumption: many natural examples in the quantum setting are reducible, the most important of which are the unitary conjugations. We note that, unlike the case of stochastic matrices, reducibility of a quantum walk is not simply described by checking the accessibility of a vertex of the graph, starting from another. Instead, one has to consider the degeneracy of the spectrum of the associated positive operator. Therefore, the importance of providing an answer to this question comes from both theoretical and applied motivations.

\medskip

{\bf Our solution:} The first element of our solution regards a spectral condition on the QMC $\Lambda=\Lambda_{T,V}$ induced by the quantum channel $T$, which is also dependent on the choice of the goal subspace $V$ of interest. Namely, if $P$ is the orthogonal projection onto $V$ and $Q=I-P$, $\mathbb{Q}=Q\cdot Q$, we define the following:

\bigskip

\underline{Assumption I}: the number 1 does not belong to the spectrum of $\mathbb{Q}\Lambda$. Equivalently, no non-null fixed point of $\Lambda$ is supported on $V^\perp$.

\bigskip

By [\cite{glv}, proof of Thm. 2.8], we have that every irreducible channel satisfies Assumption I, but there are reducible examples that also satisfy this condition. Therefore, we have a condition which is strictly weaker than irreducibility and will allow us to explore, for instance, unitary quantum walks, under the light of generalized inverses. The reason why Assumption I is a natural one is somewhat intuitive, as we wish to exclude the possibility of being trapped in the complement of the goal subspace. Nevertheless, in later sections we will provide the details regarding the use of this hypothesis in the proof or our theorem.

\medskip

The second element of our solution deals with the following issue: given a not necessarily irreducible QMC $\Lambda$, one has to clarify whether one is able to obtain a generalized inverse of $I-\Lambda$. Here, we will make use of the so-called group inverse \cite{meyer,meyer-artigo}, an object which is available in our context, and is a perfect fit to our problem. We will provide a brief but precise discussion on such inverse, noting that this construction may be of independent interest to the reader.

\medskip

As in previous works, and similarly to the classical case, the message is that generalized inverses combined with information about return times encode all the information about mean hitting times. The improvement presented in this paper is a substantial enlargement of the class of maps and subspaces to which this assertion applies.

\medskip

The contents of this work are as follows. In Section \ref{sec2} we discuss preliminaries on positive maps, probability calculations in terms of traces of such operators (the monitoring formalism), and we review the basic material on generalized inverses. In Section \ref{sec3} we review QMCs, as their simple and intuitive structure will be employed in our method for writing hitting time formulae for quantum channels. For completeness, we review the known quantum hitting time expressions in Section \ref{sec4}. In Section \ref{sec5} we define QMCs induced by quantum channels. Our main result will consist of using the group inverse so that one obtains a hitting time expression that holds for any quantum channel, with respect to any goal subspace that satisfies Assumption I. The theorem is stated in Section \ref{sec7} and its formal proof is given in the Appendix. In Section \ref{sec8} we give examples, including the setting of unitary conjugations. We present a brief summary in Section \ref{sec9}.

\section{Preliminaries}\label{sec2}

\subsection{Positive maps}

We refer the reader to \cite{bhatia} for more on the basics of positive maps. For $A \in M_n$, we write its Hilbert adjoint as $A^*$. We say that $\rho\in M_n$ is {\bf positive semidefinite} (or positive, for short), denoted by $\rho \geq 0$, when $\langle v |\rho v\rangle\geq 0$ for all $v\in \mathbb{C}^n$, where $\langle \:\cdot\:|\:\cdot\:\rangle$ is the inner product of $\mathbb{C}^n$. If $\langle v|\rho v \rangle>0$ for all $v\neq 0$, we say $\rho$ is {\bf positive definite} (or strictly positive), denoted by $\rho>0$. A matrix $\rho\in M_n$ is called a {\bf density matrix} whenever $\rho\geq 0$ and $\mathrm{Tr}(\rho)=1$. We also say that a strictly positive density is a {\bf faithful state}.

 \medskip
	
A linear map $T:M_n\to M_n$ is {\bf trace preserving} (TP) whenever $\Tr\big(T(\rho)\big)=\Tr(\rho)$, for all $\rho \in M_n$. If we have that $T(\rho)\geq 0$ whenever $\rho \geq 0$, then we say $T$ is a {\bf positive map}. A positive map $T$ is said to be {\bf completely positive} (CP) whenever it admits a Kraus decomposition, that is, if there exist $V_i\in M_n$, $i=1,\dots, k$ such that
$$T(X)=\sum_{i=1}^k V_i\rho V_i^*,\;\;\;X\in M_n.$$
A {\bf quantum channel} is a completely positive, trace preserving map (CPTP). We will have such operators in mind as we discuss hitting times, as this is one of the most studied classes of operators in the realm of quantum information theory. However, all results presented in this work also hold, more generally, to the case of positive, trace preserving maps on finite-dimensional Hilbert spaces, as the preliminary results and formal proofs do not depend on the existence of some Kraus decomposition of the positive maps involved.

\medskip

We say that the positive map $T:M_n\to M_n$ is {\bf irreducible} if the only orthogonal projections $P$ on $M_n$ such that $T$ leaves $P M_nP$ globally invariant -- that is, $T(P M_nP)\subset  PM_nP$ -- are $P=0$ and $P=I_n$, where $I_n$ is the order $n$ identity. Otherwise, we say $T$ is {\bf reducible}. We note that for positive maps acting on finite-dimensional spaces, irreducibility is equivalent to the existence of a unique faithful fixed point. 

\medskip

Every positive trace preserving map $T$ on $M_n$ has an {\bf invariant state} $w$, i.e. a density matrix such that $T(w)=w$. Then, if $T$ is irreducible, we have that $w>0$ and any other fixed point $X\in M_n$ of $T$ is a multiple of $w$, i.e. $T(X)=X$ implies $X=\lambda w$, $\lambda\in\C$, see \cite{evans}.

\medskip

If $A\in M_n$, the corresponding vector representation $vec(A)$ associated with it is given by stacking together its matrix rows. For instance, if $n=2$,
$$
A = \begin{bmatrix} A_{11} & A_{12} \\ A_{21} & A_{22}\end{bmatrix}
\quad \Rightarrow \quad 
vec(A) = \begin{bmatrix} A_{11} & A_{12} & A_{21} & A_{22}\end{bmatrix}^T.
$$
The $vec$ mapping satisfies $vec(AXB^T)=(A\otimes B)\,vec(X)$ for any square matrices $A, B, X$ \cite{hj2}. In particular, $vec(BXB^*)=vec(BX\ov{B}^T)=(B\otimes \ov{B})\,vec(X)$. More generally, this mapping allows us to consider the {\bf matrix representation} of any linear map $T$ acting on $M_n$: we define $\lceil T\rceil$ to be the matrix such that 
\beq\label{vecbas}T(X)=vec^{-1}(\lceil T\rceil \, vec(X)),\;\;\;X\in M_n.\eeq
This matrix representation will be of assistance when we perform calculations in the examples later in this work.

\medskip

We will make use of the Dirac (ket-bra) notation for vectors, that is, if $\psi\in \mathbb{C}^n$, we write $|\psi\rangle$ to represent the column vector corresponding to $\psi$ and $\langle\psi|:=(|\psi\rangle)^*$. Any matrix $A\in M_n$ may be viewed as an operator $\phi \mapsto A\phi$ on the Hilbert space $\HH_n=\C^n$ with the standard inner product. In this space, ket-bra operators also make sense: any two vectors $\phi,\psi\in\HH_n$ generate the ket-bra operator $|\phi\rangle\langle\psi|=\phi\psi^*\in M_n$. The unit vectors $\phi\in\HH_n$ are called {\bf pure states} since they may be identified up to an arbitrary phase with the states $\rho_\phi:=|\phi\rangle\langle\phi|=\phi\phi^*$. We will refer to $\HH_n$ as the {\bf pure state Hilbert space} for $M_n$.

\subsection{Probabilistic notions for positive maps}\label{subsec42}

Let $V$ denote a subspace of the pure state Hilbert space $\HH_n$ where the operators represented by the matrices of $M_n$ act. If $P$ is the orthogonal projection onto $V$ and $Q=I_n-P$, we introduce the corresponding orthogonal projections, $\mathbb{P}$ and $\mathbb{Q}$, for the space $M_n$ where the mixed states live,
$$
\mathbb{P},\mathbb{Q}:M_n\to M_n, \qquad \mathbb{P}=P\cdot P, \qquad \mathbb{Q}=Q\cdot Q.
$$
Note  that $\mathbb{P}+\mathbb{Q}+\mathbb{R}=I$, where $\mathbb{R}:=P\cdot Q+Q\cdot P$ is an orthogonal projection onto traceless matrices.

\medskip

Given a pure state $\phi\in\HH_n$, we define the following probabilistic quantities for a discrete time process whose steps follow by iterating a positive trace preserving map $T$ on $M_n$:
\begin{equation} \label{prob}
\begin{aligned}
& p_r(\phi \to V) = & & \text{probability of reaching the subspace $V$ in $r$ steps when starting at $\phi$.} 
\\
& \pi_r(\phi\to V) = & & \text{probability of reaching the subspace $V$ for the first time in $r$ steps} 
\\
& & & \text{when starting at $\phi$.} 
\\
& \pi(\phi\to V) = & & \text{probability of ever reaching the subspace $V$} 
& & \kern-100pt \text{\bf(hitting probability)} 
\\
& & & \text{when starting at $\phi$.}  
\\
& \tau(\phi\to V) = & & \text{expected time of first visit to the subspace $V$}
& & \kern-100pt \text{\bf(mean hitting time)}
\\
& & & \text{when starting at $\phi$.}  
\end{aligned}
\end{equation}
We will refer to $\phi$ and $V$ as the {\bf initial state} and the {\bf arrival (or goal) subspace} respectively. An important particular case of the subspace $V$ is the 1-dimensional space generated by a pure state $\psi\in\HH_n$. In this case we will simplify the notation by writing $\pi(\phi\to \psi)$, $\tau(\phi\to \psi)$, and so on.   

\medskip

Another particularly interesting situation arises when $\phi\in V$. This defines the so-called {\bf recurrence properties} of the process governed by the map $T$, for instance, the probability $\pi(\phi\to V)$ of the return of $\phi$ to the subspace $V$ {\bf (return probability)}, and the expected time $\tau(\phi\to V)$ of such a return {\bf (mean return time)}.

\medskip

According to the monitoring approach, the above notions can be expressed in terms of appropriate trace expressions involving maps on $M_n$ acting on the density matrix $\rho_\phi=|\phi\>\<\phi|$ of the initial state:
\begin{flalign}
& p_r(\phi\to V) = \Tr(\mathbb{P}T^r \rho_\phi), 
\nonumber \\
& \pi_r(\phi\to V) = \Tr(\mathbb{P}T(\mathbb{Q}T)^{r-1} \rho_\phi), 
\nonumber \\
& \pi(\phi\to V) = \sum_{r\geq 1} \pi_r(\phi\to V) = 
\sum_{r\geq 1} \Tr(\mathbb{P}T(\mathbb{Q}T)^{r-1} \rho_\phi), 
\label{probPOS}  \\
&\tau(\phi\to V)=
\begin{cases}
\infty, 
& \text{ if } \pi(\phi\to V)<1,  
\\
\displaystyle \sum_{r\geq 1} r\pi_r(\phi\to V) = 
\sum_{r\geq 1} r\Tr(\mathbb{P}T(\mathbb{Q}T)^{r-1} \rho_\phi), 
& \text{ if } \pi(\phi\to V)=1. 
\end{cases} 
\nonumber 
\end{flalign}
Here and in what follows we write for convenience $T\rho$ to denote $T(\rho)$ for any $\rho\in M_n$. Since $I-\mathbb{Q}=\mathbb{P}+\mathbb{R}$ differs from $\mathbb{P}$ in an orthogonal projection $\mathbb{R}$ onto traceless matrices, the traces appearing in \eqref{probPOS} may be equivalently expressed by substituting $\mathbb{P}$ by $I-\mathbb{Q}$. 

\medskip

Given a trace preserving, positive map $T$ on $M_n$, it turns out that $\pi(\phi\to V)=1$ for every pure state $\phi$ and every subspace $V$ satisfying Assumption I (the proof given in [\cite{glv}, Prop. 2.7] for CPTP maps also holds for positive TP maps). Therefore, in this case the expression of the mean hitting time is the one in the bottom line of \eqref{probPOS}.

\medskip

In order to discuss hitting time formulae, we define the function
\beq\label{gf11}
\mathbb{G}(z) = 
\sum_{r\geq 1} T(\mathbb{Q}T)^{r-1}z^r = 
zT(I-z\mathbb{Q}T)^{-1}, 
\qquad |z|<1,
\eeq
which is analytical since, with respect to the operator norm, $\|\mathbb{Q}\|=1$ as an orthogonal projection, while the Russo-Dye theorem  \cite{bhatia} implies that $\|T\|=\|T^*\|=\|T^*(I_n)\|=\|I_n\|=1$ since $T$ is trace preserving. The interest of this function relies on the fact that, according to \eqref{probPOS}, its boundary behaviour around $z=1$ provides hitting probabilities and mean hitting times,
\begin{equation} \label{HKlim}
\begin{aligned} 
& \pi(\phi\to V) = \Tr(\mathbb{P}H\rho_\phi) = \Tr((I-\mathbb{Q})H\rho_\phi),
& \qquad & H:=\lim_{x\uparrow 1}\mathbb{G}(x),
\\
& \tau(\phi\to V) = \Tr(\mathbb{P}K\rho_\phi) = \Tr((I-\mathbb{Q})K\rho_\phi),
& & K:=\lim_{x\uparrow 1}\mathbb{G}'(x).
\end{aligned}
\end{equation}
As limits of positive maps, $H$ and $K$ are positive maps too. We will refer to $H$ and $K$ as the {\bf hitting probability map} and the {\bf mean hitting time map} for the subspace $V$. The existence of these limits, explained in \cite{cv}, is a direct consequence of Assumption I, and is an important point of discussion later in this work.

\medskip

We write the following useful expressions: under Assumption I (stated in Section \ref{intro3}), we have the analyticity of $\mathbb{G}(z)$ at $z=1$, which allows us to express  
\begin{equation}\label{HK}
\begin{aligned}
& H = \mathbb{G}(1) = T(I-\mathbb{Q}T)^{-1},
\\
& K = \mathbb{G}'(1) = 
T(I-\mathbb{Q}T)^{-1} + 
T\mathbb{Q} T (I-\mathbb{Q}T)^{-2} =
T(I-\mathbb{Q}T)^{-2}.
\end{aligned}
\end{equation}

\medskip

Now consider the following $V$-dependent {\bf block matrix representation} of a map $T \colon M_n \to M_n$:
\beq\label{phirep}
T = 
\begin{bmatrix} T_{11} & T_{12} \\ T_{21} & T_{22}\end{bmatrix} :=
\begin{bmatrix} 
(I-\mathbb{Q})T(I-\mathbb{Q}) & (I-\mathbb{Q})T\mathbb{Q} \\ 
\mathbb{Q}T(I-\mathbb{Q}) & \mathbb{Q}T\mathbb{Q}
\end{bmatrix}.
\eeq

\medskip

The block matrix representation of the positive maps $H$ and $K$ yields information about the first visit to the subspace $V$, when starting at states supported on $V$ or $V^\bot$. More precisely, from \eqref{HKlim} we find that
\begin{align} \label{piH}
& \begin{aligned}
& \pi(\phi\to V) = 
\begin{cases} 
\Tr((I-\mathbb{Q})H(I-\mathbb{Q})\rho_\phi) = \Tr(H_{11}\rho_\phi), 
& \text{if} \; \phi\in V, 
\\
\Tr((I-\mathbb{Q})H\mathbb{Q}\rho_\phi) = \Tr(H_{12}\rho_\phi), 
& \text{if} \; \phi\in V^\bot,
\end{cases}
\qquad H=\begin{bmatrix} H_{11} & H_{12} \\ H_{21} & H_{22} \end{bmatrix},
\end{aligned}
\\ \label{tauK}
& \begin{aligned}
& \tau(\phi\to V) = 
\begin{cases} 
\Tr((I-\mathbb{Q})K(I-\mathbb{Q})\rho_\phi) = \Tr(K_{11}\rho_\phi), 
& \text{if} \; \phi\in V, 
\\
\Tr((I-\mathbb{Q})K\mathbb{Q}\rho_\phi) = \Tr(K_{12}\rho_\phi), 
& \text{if} \; \phi\in V^\bot,
\end{cases}
\qquad K=\begin{bmatrix} K_{11} & K_{12} \\ K_{21} & K_{22} \end{bmatrix}.
\end{aligned}
\end{align}
In particular, the first diagonal blocks of $H$ and $K$ provide recurrence properties: $H_{11}$ and $K_{11}$ give respectively the probabilities and mean times for the return of a state $\phi\in V$ to the subspace $V$ where it lies. We will call $H_{11}$ and $K_{11}$ the {\bf return probability map} and the {\bf mean return time map} for the subspace $V$, respectively.

\subsection{Generalized inverses}\label{sec43} The reader interested in more details on generalized inverses may consult \cite{meyer}, also see \cite{hunter,oqwmhtf2} and references therein. Here we review only the necessary facts for our results.

		\begin{definition}\label{ginvdefinition}
		A {\bf generalized inverse}, or a {\bf g-inverse} of a matrix $A\in M_n$, is any matrix $A^-$ such that 
		\begin{align*}
			AA^-A=A.
		\end{align*}
	\end{definition}
We can see that if a matrix $A$ is invertible, then $ A^-=A^{-1} $ is the only g-inverse of $A$.  In this work we will be mostly interested in a particular kind of g-inverse. In order to study it, we define the following.
	
	\begin{definition}\label{indexdefinition}
		Given a matrix $A\in M_n$, the smallest nonnegative integer $k$ such that $\Ran(A^k)=\Ran(A^{k+1})$ is the {\bf index} of $A$, denoted by $\Ind(A)=k$.
	\end{definition}
	
	We consider $A^0$ to be the identity matrix for any matrix $A$. So, for example, nonsingular matrices are precisely those with index zero.
	
	\begin{definition}
		If $A\in M_n$ with $\Ind(A)=k$, and if $A^D\in M_n$ is a matrix such that 
		\begin{enumerate}
			\item $A^DAA^D=A^D$
			\item $A^DA=AA^D$
			\item $A^{k+1}A^D=A^k$
		\end{enumerate} then $A^D$ is called the {\bf Drazin inverse} of $A$.
	\end{definition}
	
	It can be shown that the Drazin inverse of any square matrix always exists and is unique \cite{meyer}. We also note, however, that the Drazin inverse is not a g-inverse in general. The next theorem [\cite{meyer}, Thm. 7.2.4] specifies when that happens.
	
	\begin{theorem}\cite{meyer}\label{DrazinGroup}
		Let $A\in M_n$. Then $AA^DA=A$ if, and only if $\Ind(A)\leq 1$.
	\end{theorem}
	
	A Drazin inverse which is also a g-inverse receives a special name:
	
	\begin{definition}
		Let $A\in M_n$ with $\Ind(A)\leq 1$. Then the Drazin inverse of $A$ is denoted by $A^\#$ and called the {\bf group inverse} of $A$.
	\end{definition}
	We could have alternatively defined the group inverse of a square matrix $A$ as the unique matrix $A^\#$, when it exists, satisfying the three following conditions:
	\begin{align*}
		A^\#AA^\#=A^\#, \quad A^\#A=AA^\#, \quad AA^\#A=A.
	\end{align*}
	
	Therefore, the group inverse of a given square matrix exists, by Theorem \ref{DrazinGroup}, precisely when its index is not greater than 1. This existence of such inverse will allows us, in particular, to study reducible quantum dynamics as we will see later in this work.

\medskip

By [\cite{meyer}, Thm. 7.6.1], we can calculate the Drazin inverse of $A$ via the limit
	 $$A^D =\lim_{z\to 0}(A^{l+1}+zI)^{-1}A^l,$$ 
	 for every integer $l\geq \Ind(A)=1$. In particular, we can calculate it by taking $l=1$ and obtain
\beq\label{drazinmeth}
		A^D = \lim_{z\to 0}(A^2+zI)^{-1}A.
\eeq

In order to establish the existence of the group inverse in our setting, we recall:

\begin{theorem}\label{wolfproposition}[\cite{wolf}, Prop. 6.2] {\bf(Trivial Jordan blocks for peripheral spectrum).} Let $T:M_n\rightarrow M_n$ be a trace preserving (or unital) positive linear map. If $\lambda$ is an eigenvalue of $T$ with $|\lambda|=1$, then its geometric multiplicity equals its algebraic multiplicity, i.e., all Jordan blocks for $\lambda$ are one-dimensional.
\end{theorem}
	
Then, we have the following:
	
	\begin{pro}\label{groupinverseexists} For every positive, trace preserving map $T$ on a finite-dimensional Hilbert space, the map $A^\#$ exists, where $A=I-T$.
	\end{pro}
	\begin{proof} By Theorem \ref{wolfproposition}, the matrix representation for $T$ has only one-dimensional Jordan blocks relative to the eigenvalue 1, so its Jordan decomposition will have, for some choice of matrix $X$, the form
	\begin{align*}
		 T = X\begin{bmatrix}
			I &  \\  & B
		\end{bmatrix}X^{-1},
	\end{align*} where $I$ is an identity matrix of some dimension and $B$ is a matrix with no eigenvalue equal to 1. It follows that \begin{align*}
		A=I- T = X\begin{bmatrix} O & \\ & I-B \end{bmatrix}X^{-1} \quad \Longrightarrow \quad A^2 = X\begin{bmatrix}
			O & \\ & (I-B)^2
		\end{bmatrix}X^{-1}.
	\end{align*} As $B$ has no eigenvalues equal to 1, it follows that $I-B$ is nonsingular, and hence so is $(I-B)^2$. Then it is clear that $A$ and $A^2$ have the same rank as $I-B$ and $(I-B)^2$, so $\Ind(A)\leq 1$. Therefore, by Theorem \ref{DrazinGroup}, the group inverse $A^\#$ exists. 
	
	\end{proof}

Finally, a basic result relating ergodic means and the group inverse is the following. The proof is given in the Appendix.

\begin{lemma}\label{mnov_lemma} For every positive, trace preserving map $T$ on a finite-dimensional Hilbert space,
\beq\label{limtt1}\lim_{n\to\infty}\frac{I+T+T^2+\cdots+T^{n-1}}{n}=I-A^\#A,\;\;\;\;\;\;A=I-T.\eeq
Moreover, given any density $\rho$, we have that $(I-A^\#A)\rho$ is a density which is invariant for $T$.
\end{lemma}


\section{A special class of channels: quantum Markov chains}\label{sec3}

Let $\Phi$ be the following matrix of operators,
\begin{align}
		\Phi=\begin{bmatrix} \Phi_{11} & \cdots & \Phi_{1n}\\
			\Phi_{21} &\cdots & \Phi_{2n}\\
			\vdots & \ddots & \vdots \\
			\Phi_{n1} &\cdots &\Phi_{nn}\\
		\end{bmatrix},
	\end{align}
so that each $\Phi_{ij}:M_k\to M_k$ is a positive map. If we define the set of densities on $n$ vertices and $k$ internal degrees of freedom as
	\begin{align*}
		\mathcal D_{n,k}:=\left\{ \rho=\begin{bmatrix} \rho_1\\
			\rho_2\\
			\vdots\\
			\rho_n
		\end{bmatrix} \; : \; \rho_i\in M_k,\;\; \rho_i\geq 0,\; i=1,\dots,n,\;\; \sum_{j=1}^n \Tr(\rho_j)=1\right\},
	\end{align*}
we define the action of $\Phi$ over $D_{n,k}$ as
	\begin{align}\label{qmcbas00}
		\Phi(\rho)=\begin{bmatrix} \Phi_{11} & \cdots & \Phi_{1n}\\
			\Phi_{21} &\cdots & \Phi_{2n}\\
			\vdots & \ddots & \vdots \\
			\Phi_{n1} &\cdots &\Phi_{nn}\\
		\end{bmatrix}\cdot \begin{bmatrix} \rho_1\\ \rho_2 \\ \vdots \\ \rho_n
		\end{bmatrix} := \begin{bmatrix} \sum_{j=1}^n \Phi_{1j}(\rho_j) \\
			\sum_{j=1}^n \Phi_{2j}(\rho_j)\\
			\vdots\\
			\sum_{j=1}^n \Phi_{nj}(\rho_j)
		\end{bmatrix}.
	\end{align}
If we have that $\Phi$ preserves $\mathcal{D}_{n,k}$, that is, $\mathrm{Tr}(\Phi(\rho))=\mathrm{Tr}(\rho)=\sum_i \mathrm{Tr}(\rho_i)=1$, $\rho\in \mathcal{D}_{n,k}$, we say that $\Phi$ is a {\bf quantum Markov chain (QMC)} acting on $n$ vertices with internal degree of freedom $k$. This definition is essentially due to S. Gudder \cite{gudder} except that in such work the $\Phi_{ij}$ are required to be completely positive. We will also consider the cases for which only positivity is assumed, as this simpler assumption extends the applicability of the theory.

\medskip

In the special case for which the $\Phi_{ij}$ are given simply by $\Phi_{ij}(\rho) = B_{ij}\rho B_{ij}^*$,  $B_{ij}\in M_k$  the QMC reduces to what we call an {\bf open quantum walk} (OQW), following Attal et al. \cite{attal}, and the map $\Phi$ can be written as
$$\Phi(\rho) = \sum_{i,j} M_{ij}\rho M_{ij}^*\;,\quad M_{ij}=B_{ij}\otimes |i\rangle \langle j|,$$ 
the above summations ranging over the set of vertices. Then, the preservation of trace is equivalent to 
$$\sum_{i} B_{ij}^*B_{ij} = I,\;\;\;\forall j.$$
In this case, and also more generally, density matrices in $\mathcal{D}_{n,k}$ can be written as $\rho=\sum_{i} \rho_i\otimes|i\rangle$, and we say $\rho$ is a direct sum of positive matrices $\rho_i\in M_k$, each of them concentrated on vertex $i$.
	
	\medskip
	
	We can think of our system as a simulation of a Markov chain on a graph with a set of $n$ \textbf{vertices}, or \textbf{sites}, but the particle performing the walk also has an internal state represented by a linear operator on a Hilbert space of dimension $k$. QMCs and OQWs can be defined in more general complex Hilbert spaces where we do not have necessarily finitely many vertices and degrees of freedom, but our focus in this work will be on finite systems.
	
	\medskip
	
	\begin{figure}[ht]
		
		\begin{adjustbox}{center}
			\begin{tikzpicture}
				[->,>=stealth',shorten >=1pt,auto,node distance=2.0cm,
				semithick]
				\node[main node] (1) {$1$};
				\node[main node] (2) [right=2.0cm and 2.0cm of 1,label={[]$$}]  {$2$};
				
				\node[main node] (3) [below=2.0cm and 2.0cm of 2,label={[]$$}]  {$3$};
				\node[main node] (4) [left = 2.0cm and 2.0cm of 3] {$4$};


				
				

				
				
				\path[draw,thick]
				
				(1) edge   [bend left]  node {$\Phi_{21}$}    (2)

				(2) edge   [bend left]     node[above]{$\Phi_{12}$} (1)
				
				(2) edge   [loop right]     node {$\Phi_{22}$} (2)
				
				(2) edge   [bend left]     node {$\Phi_{32}$} (3)
				(3) edge   [bend left]     node[right] {$\Phi_{23}$} (2)
				
				(1) edge   [bend left]     node[left] {$\Phi_{41}$} (4)
				(4) edge   [bend left]     node[left] {$\Phi_{14}$} (1)
				
				(4) edge   [bend left]     node[below] {$\Phi_{34}$} (3)

				
				
				;
			\end{tikzpicture}
		\end{adjustbox}
		\caption{Schematic illustration of a QMC on 4 vertices. The walk is realized on a graph with a set of vertices denoted by $1,2,3,4$ and each operator $\Phi_{ij}$ is a positive map describing a transformation in the internal degree of freedom of the particle during the transition from vertex $j$ to vertex $i$. In the particular case that for every $i,j$, $\Phi_{ij}=B_{ij}\cdot B_{ij}^*$ for certain matrices $B_{ij}$, we say the QMC is an open quantum walk (OQW).}
	\end{figure}
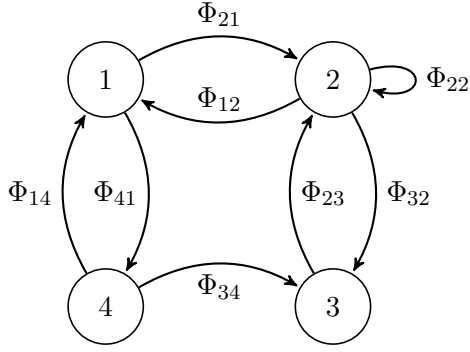\

	For a density $\rho\in\mathcal{D}_{n,k}$, we establish the correspondence
	\begin{align*}
		\rho = \begin{bmatrix}
			\rho_1\\ \vdots \\ \rho_n
		\end{bmatrix} \quad \longleftrightarrow \quad |\rho\rangle :=\begin{bmatrix}
			|\rho_1\rangle \\ \vdots \\ |\rho_n\rangle
		\end{bmatrix} = \begin{bmatrix}
			\vect{(\rho_1)} \\ \vdots \\ \vect{(\rho_n)}
		\end{bmatrix} \in \mathbb C^{nk^2},
	\end{align*}
	where we define $|\rho_i\rangle := \vect{(\rho_i)} \in \mathbb C^{k^2}$. Then, if we define the matrix
	\begin{align*}
		\ceil {\Phi} := \begin{bmatrix}
			\ceil{\Phi_{11}} & \cdots & \ceil{\Phi_{1n}}\\
			\vdots & \ddots & \vdots \\
			\ceil{\Phi_{n1}} & \cdots & \ceil{\Phi_{nn}}
		\end{bmatrix} \in M_{nk^2},
	\end{align*} we can express equation (\ref{qmcbas00}) as 
	\begin{align*}
		\ceil{\Phi}|\rho\rangle=\begin{bmatrix} \ceil{\Phi_{11}} & \cdots & \ceil{\Phi_{1n}}\\
			\ceil{\Phi_{21}} &\cdots & \ceil{\Phi_{2n}}\\
			\vdots & \ddots & \vdots \\
			\ceil{\Phi_{n1}} &\cdots &\ceil{\Phi_{nn}}\\
		\end{bmatrix}\cdot \begin{bmatrix}| \rho_1\rangle\\ |\rho_2\rangle \\ \vdots \\ |\rho_n\rangle
		\end{bmatrix} = \begin{bmatrix} \sum_{j=1}^n \ceil{\Phi_{1j}}|\rho_j\rangle \\
			\sum_{j=1}^n \ceil{\Phi_{2j}}|\rho_j\rangle\\
			\vdots\\
			\sum_{j=1}^n \ceil{\Phi_{nj}}|\rho_j\rangle
		\end{bmatrix}, 
	\end{align*}
	which in fact corresponds to the usual multiplication of matrices and vectors: these matrices are partitioned in $n^2$ blocks of order $k^2$, and the vectors are partitioned in $n$ vectors of length $k^2$.

	\medskip
	
	In the case of finite, irreducible QMCs, the iterates $\ceil{\Phi}^m$ converge to $|w\rangle \langle e_{I_k^n}|$ when $m$ goes to infinity \cite{oqwmhtf2}, where 
	\begin{align*}
		|e_{I_k^n}\rangle:=\begin{bmatrix}
			\vect{(I_k)} \\ \vect{(I_k)} \\ \vdots \\ \vect{(I_k)}
		\end{bmatrix} =\begin{bmatrix}
			|e_{I_k}\rangle  \\ |e_{I_k}\rangle \\ \vdots \\ |e_{I_k}\rangle
		\end{bmatrix}\in \mathbb C^{nk^2}, \quad |e_{I_k}\rangle := \vect{(I_k)} \in \mathbb C^{k^2},
	\end{align*} and $|w\rangle$ is the vector form of the fixed density $w$ of the QMC. We also write $|e_{I_k^n}\rangle$ only as $|e_I\rangle$ for simplicity. We call $\Omega:=|w\rangle \langle e_{I_k^n}|$ the {\bf fixed map} associated with the irreducible QMC $\Phi$.
	
	\medskip
	
	We note that the vec function establishes a unitary equivalence between the Hilbert spaces $M_n$ and $\mathbb C^{n^2}$  with their inner products $\langle\; \cdot \mid \cdot\;\rangle_{M_n}$ and $\langle \; \cdot \mid \cdot \; \rangle_{\mathbb C^{n^2}}$ \cite{cv} by the fact that $$\langle B \mid A\rangle_{M_n}=\operatorname{Tr}\left(B^* A\right)=\sum_{i, j} \overline{B_{i j}} A_{i j}=\operatorname{vec}(B)^* \operatorname{vec}(A)=\langle\operatorname{vec}(B) \mid \operatorname{vec}(A)\rangle_{\mathbb{C}^{n^2}}, \quad A, B \in M_n.$$ With this we have that for $\rho_i \in M_k$,
	\begin{align*}
		\Tr(\rho_i) = \langle e_{I_k}|\rho_i \rangle,
	\end{align*} and for $\rho\in \mathcal{D}_{n,k}$, 
	\begin{align*}
		\Tr(\rho) = \sum_i\Tr(\rho_i) = \sum_i \langle e_{I_k} | \rho_i\rangle = \langle e_{I_k^n} | \rho \rangle.
	\end{align*}

\section{A review: hitting time formulae for the irreducible case}\label{sec4}

For completeness of exposition we review some recent results. The following can be seen as a quantum version of the classical MHTF (\ref{clmhtf}) stated in Section \ref{intro2}. For the statement below, one should recall the notation used in (\ref{phirep}) and (\ref{tauK}), regarding the matrix representation of operators. If $T:M_n\to M_n$ is an irreducible, positive, trace preserving map acting on $M_n$, let $w$ be the unique strictly positive density matrix which is invariant for $T$ and define the fundamental map $Z=(I-T+\Omega)^{-1}$, where $\Omega=|w\rangle\langle I_n|$.

\begin{theorem}\cite{cv}\label{mhtf} 
{\bf (MHTF for irreducible positive maps and orthogonal states).} Let $T: M_n\to M_n$ denote an irreducible, positive, trace preserving map and $V$ a nontrivial subspace of the pure state Hilbert space $\HH_n$. If $Z$ is the associated fundamental map, $K$ is the mean hitting time map \eqref{HK} and $D=K_d=diag(K_{11}, K_{22})$ is its diagonal part, we have that for every $\psi\in V$ and $\phi\in V^\perp$,
\beq \label{mhtff1}
\tau(\phi\to V) = \Tr\Big((DZ)_{11}\rho_\psi-(DZ)_{12}\rho_\phi\Big)
= \Tr\Big(K_{11}(Z_{11}\rho_\psi-Z_{12}\rho_\phi)\Big).
\eeq
In particular, the quantity $\Tr((DZ)_{11}\rho_\psi)$ is independent of the choice of $\psi\in V$.
\end{theorem}

It is worth recalling that the above result is also valid if the initial state and goal subspace are not orthogonal, and this is obtained in terms of a conditioning on the first step: the mean time equals 1 plus the mean time to reach $V$ given that we are at the complement of $V$ after one step, see [\cite{cv}, Lemma 4.6 and Thm. 4.7].

\medskip

In a related construction, the KSMH formula makes use of generalized inverses in order to produce hitting time expressions. In order to state it properly, first we review a lemma proven in \cite{oqwmhtf2} in the setting of QMCs. The result is inspired by Hunter \cite{hunter} and provides a characterization for any g-inverse.

	\begin{lemma}\cite{hunter,oqwmhtf2}\label{generalg}
		Let $\Phi$ be an irreducible QMC on a finite graph with stationary density $w$. Let $\ket t, \ket u \in \mathbb C^{nk^2}$ be such that $\bracket{e_I}{t}  \neq 0$ and $\bracket{u}{w} \neq 0$. Then any g-inverse of $I-\Phi$ can be written as:
		\begin{align*}
			G=\big(I-\lceil \Phi\rceil +\ket t \bra u\big)^{-1} +\ket w \bra f + \ket g \bra{e_I},
		\end{align*}
		where $\ket f, \ket g$ are arbitrary vectors.
	\end{lemma}
	
\medskip

Regarding the hitting time operator for QMCs, we follow the notation given in \cite{oqwmhtf2} for generating functions. Namely, we write an expression that resembles (\ref{gf11}) for positive maps, by defining for any vertices $i, j$, 
\beq\label{gf12}
\mathbb{G}_{ij}(z) = 
\sum_{r\geq 1} \mathbb{P}_iT(\mathbb{Q}_iT)^{r-1}\mathbb{P}_jz^r = 
z\mathbb{P}_iT(I-z\mathbb{Q}_iT)^{-1}\mathbb{P}_j, 
\qquad |z|<1,
\eeq
where $\mathbb{P}_i$ is the projection on the $i$-th vertex of the graph associated with the QMC, and $\mathbb{Q}_i$ is the projection onto the vertices distinct from $i$. Then, we define $K:=(K_{ij})$, where
\beq\label{kqmcdef} 
K_{ij}=\lim_{x\uparrow 1} \frac{d}{dx} x\mathbb{G}_{ij}(x).\eeq
With such notation, we may state the following:

\medskip

	\begin{theorem}\label{discreteMHTF2019}\cite{oqwmhtf2} {\bf (KSMH formula for irreducible QMCs).}
		Let $\Phi$ be an irreducible QMC on a finite graph with $n \geq 2$ vertices and $k\geq 2$ internal degrees of freedom,  let $\Omega$ be its fixed map. Let $K=(K_{ij})$ denote the matrix of mean hitting time operators to vertices $i=1,\dots,n$, $D=K_d=\diag(K_{11},\dots,K_{nn})$. Let $G$ any g-inverse of $I-\Phi$, and let $E$ denote the block matrix for which each block equals the identity of order $k^2$. (a) The mean hitting time for the walk to reach vertex $i$, beginning at vertex $j$ with initial density $\rho_j$ is given by
		\begin{align*}
			\tau(\rho_j\otimes|j\rangle \rightarrow |i\rangle) = \Tr(K_{ij}\rho_j) = \Tr\Big(\big[D\big(\Omega G - (\Omega G)_dE +I-G+G_dE\big)\big]_{ij}\rho_j\Big),\;\;\;\textrm{ for all } i,j=1,\dots,n.
		\end{align*}
		
		(b) By setting $G= (I-\lceil\Phi\rceil+|u\rangle\langle e_I|)^{-1}+|f\rangle \langle e_I|$, with $|f\rangle$ arbitrary, and $|u\rangle$ such that $\langle e_I|u\rangle \neq 0$, we have that for every vertex $i$ and initial density $\rho_j$ on vertex $j$, 
		\begin{align}\label{hunterf11}
			\tau(\rho_j\otimes |j\rangle \rightarrow |i\rangle) = \Tr(K_{ij}\rho_j)=\Tr\Big(\big[D(I-G+G_dE)\big]_{ij}\rho_j\Big),\;\;\;\textrm{ for all } i,j=1,\dots,n.
		\end{align}
	\end{theorem}

Expression (\ref{hunterf11}) above should be seen as a quantum version of the classical equation (\ref{ht1}) stated in Section \ref{intro2}. We note that in reference \cite{oqwmhtf2}, the above result is called Hunter's formula for QMCs, due to the proximity of the proof of the result with the one presented in \cite{hunter} (also see the remarks right after equation (\ref{ht1})). For simplicity, in this work we will focus on generalizing expression (\ref{hunterf11}). We call the operator in brackets inside the trace in (\ref{hunterf11}) the  {\bf KSMH kernel} associated with the hitting time formula.

\medskip

\begin{remark} The fundamental matrix $Z$ is a particular generalized inverse that falls under the conditions specified by the above theorem, therefore we must have that the mean hitting time to reach vertex $i$ from $j$ with initial density $\rho_j$ can be calculated in terms of such inverse. We have:
	\begin{align}\label{relationmhtf}
		\Tr\big(K_{ij}\rho_j\big) &= \Tr\Big(\big[D\big(I-G+G_dE\big)\big]_{ij}\rho_j\Big) \nonumber\\
		&= \Tr\big(D_{ij}\rho_j\big)-\Tr\big((DZ)_{ij}\rho_j\big)+\Tr\Big(\big(D(Z_dE)\big)_{ij}\rho_j\Big).
	\end{align}
Now, note that 
	\begin{align*}
		\big(D(Z_dE)\big)_{ij}&=\sum_{l=1}^n\sum_{m=1}^nD_{il}(Z_d)_{lm}E_{mj} = \sum_{l=1}^nD_{il}Z_{ll}\\
		&=K_{ii}Z_{ii} = (DZ)_{ii},
	\end{align*} so we can substitute this on equation (\ref{relationmhtf}) and rearrange terms to obtain 
	\begin{align*}
		&\Tr\big(K_{ij}\rho_j\big)-\Tr\big(D_{ij}\rho_j\big) = -\Tr\big((DZ)_{ij}\rho_j\big)+\Tr\big((DZ)_{ii}\rho_j\big)\\
		\Longrightarrow & \Tr\big((K-D)_{ij}\rho_j)=\Tr \big(\big[(DZ)_{ii}-(DZ)_{ij}]\rho_j\big),
	\end{align*} which is the formula given by Theorem \ref{mhtf} when $i\neq j$. 
	
	\end{remark}

	\section{QMC induced by a quantum channel}\label{sec5}
	
Let $T:M_n\to M_n$ be any quantum channel. For $V\subset \mathbb{C}^n$ subspace, let $P\in M_n$ be the orthogonal projection onto $V$ and $Q=I_n - P$. We define the operators $\mathbb P$ and $\mathbb Q$ on $M_n$ by $\mathbb PX = PXP$ and $\mathbb QX = QXQ$, respectively, for $X\in M_n$. For instance, we can take $\psi$ to be any pure state, so that $V$ is the 1-dimensional subspace spanned by $\psi$, which implies that $P=|\psi\>\<\psi|$. 

\medskip

Regarding the problem of obtaining a KSMH formula for $T$, we propose the following idea: we form a new map $\Lambda=\Lambda_{T,V}$ dependent on both the quantum channel $T$ and the arrival subspace $V$, given by the following matrix of operators,
	\begin{align}\label{LambdaChannel}
		\Lambda=\Lambda_{T,V} :=\begin{bmatrix} \mathbb{P}T & \mathbb{P}T \\ \mathbb{Q}T & \mathbb{Q}T \end{bmatrix}.
	\end{align}  
	

	 Note that to each of the subspaces $V$ and $V^\perp$ on which $T$ acts we have a site correspondence with respect to $\Lambda$: as in the first row we have the projection $\mathbb{P}$ appearing on the left of $T$, site 1 corresponds to $V$ and vectors of the form $[ X \; 0]^T$; analogously, in the second row $\mathbb{Q}$ appears on the left of $T$, so site 2 corresponds to $V^\perp$ and vectors of the form $[0\; X]^T$, for $0, X\in M_n$. 
	 
	
	\medskip
	
	We also note that although $\Lambda$ is not equal to operator $T$, it will be an essential element for what follows. The importance of such map is explained by the following lemma:

	\begin{lemma} Let $T:M_n\rightarrow M_n$ be a quantum channel, $V\subset \mathbb C^n$ a subspace. Let $P$ be the orthogonal projector onto $V$, let $Q=I-P$ and $\mathbb Q:= Q \cdot Q$. Let $\Lambda$ be defined as in (\ref{LambdaChannel}). We have:
		\begin{enumerate}
					\item The mean hitting time for $T$ to reach subspace $V$ starting from a state $\phi \in V^\perp$ is the same as the mean hitting time for $\Lambda$ starting at site $|2\rangle$, with initial density $\rho_\phi=|\phi\rangle \langle \phi|$, to reach state $|1\rangle$.
			\item $\Lambda$ is a QMC on two vertices. In particular, it is a positive, trace preserving map.
		\end{enumerate}
	\end{lemma}	
	\begin{proof} 1. We have, noting that $\mathbb{Q}_1=\mathbb{P}_2$ (the complement of vertex 1 equals vertex 2),
	\begin{align}\label{meanhittingtimescoincide}
		\nonumber \tau_\Lambda(\rho_\phi\otimes|2\rangle \rightarrow |1\rangle)&:= \sum_{r\geq 1}r\Tr(\mathbb P_1 \Lambda (\mathbb Q_1\Lambda)^{r-1}\mathbb{P}_2\rho_\phi)\\
		\nonumber  &=\sum_{r\geq 1} r\Tr(\mathbb P_1\Lambda \mathbb{P}_2(\mathbb P_2 \Lambda \mathbb P_2)^{r-1} \rho_\phi)\\
		 &= \sum_{r \geq 1} r\Tr(\mathbb P T(\mathbb QT)^{r-1}\rho_\phi)=\tau_T(\phi\rightarrow V),\\ \nonumber
	\end{align}	
	which is the mean hitting time for the quantum channel, provided the condition $\pi(\phi\rightarrow V)=1$. 
	
	\medskip
	
	2. Map $\Lambda$, is a block matrix where each element is a positive operator, so it corresponds to a positive map acting on a direct sum of matrices. Moreover, as $T$ is trace preserving, we have that
	$$\rho=\begin{bmatrix} \rho_1 \\ \rho_2 \end{bmatrix}\;\Longrightarrow\; \Lambda(\rho)=\begin{bmatrix} \mathbb{P}T & \mathbb{P}T \\ \mathbb{Q}T & \mathbb{Q}T\end{bmatrix}\begin{bmatrix} \rho_1 \\ \rho_2 \end{bmatrix}=\begin{bmatrix} \mathbb{P}T(\rho_1+\rho_2) \\ \mathbb{Q}T(\rho_1+\rho_2) \end{bmatrix}$$
	$$\Longrightarrow \mathrm{Tr}(\Lambda(\rho))=\mathrm{Tr}(\mathbb{P}T(\rho_1+\rho_2))+\mathrm{Tr}(\mathbb{Q}T(\rho_1+\rho_2))=\mathrm{Tr}\Big[(\mathbb{P}+\mathbb{Q})T(\rho_1+\rho_2)\Big]$$
	$$=\mathrm{Tr}(T(\rho_1+\rho_2))=\mathrm{Tr}(T(\rho_1))+\mathrm{Tr}(T(\rho_2))=\mathrm{Tr}(\rho_1)+\mathrm{Tr}(\rho_2)=\mathrm{Tr}(\rho),$$	
	noting that, for the third trace equality above, $\mathbb{P}+\mathbb{Q}=I-\mathbb{R}$, where $\mathbb{R}$ is the projection onto the space of traceless matrices. Therefore, $\Lambda$ is a QMC on two vertices.
	
	\end{proof}
	
As stated by the lemma above, map $\Lambda$ is in fact a QMC: we call it the {\bf QMC induced by channel $T$ and subspace $V$}. Once again, the above reasoning also applies to any positive, trace preserving map.

\section{Beyond the irreducible case: using the group inverse}\label{sec7}

Our review of results for irreducible QMCs, together with the notion of QMC induced by a quantum channel leads us to some  questions. On one hand, if for a given channel one is able to ensure that the induced QMC is irreducible, then one has
Theorem \ref{discreteMHTF2019}, so that hitting times can be calculated with generalized inverses. However, irreducibility may be quite a strong request, generally speaking. One of our main goals in this work is to find situations for which one is able to dismiss irreducibility, by considering a weaker condition. In the irreducible case, the existence of generalized inverses for QMCs is well-known. With this in mind, we ask the following natural question: are there g-inverses that can be obtained even in the reducible case? We have a positive answer to this, which is in fact given by the group inverse (this is the content of Proposition \ref{groupinverseexists}). Below we state our main result, for which we do not assume irreducibility, and in fact is valid for any positive, trace preserving map. The proof can be seen in the Appendix.
	
	\begin{theorem}\label{BeyondIrreducibleTheorem}
		Let $T: M_n\to M_n$ be a positive, trace preserving map and let $V$ be a goal subspace. Let $\Lambda=\Lambda_{T,V}$ be the QMC induced by $T$ and $V$, let $P$ be the orthogonal projection onto $V$, $Q=I-P$ and let $\mathbb{Q}=Q\cdot Q$. Suppose that $\Lambda$ and $V$ satisfy Assumption I. Let $A=I-\Lambda$, let $A^\#$ denote its group inverse and write $A_d^\#=diag((A^\#)_{11},(A^\#)_{22})$. Let $D=\diag(K_{11},K_{22})$, where $K$ is the  mean hitting time map (\ref{HK}) with respect to $V$, and let $E$ denote the block matrix for which each block equals the identity.	Then, the mean time to reach subspace $V$ under the action of $T$ is given by
		$$\tau(\rho\to V)=\mathrm{Tr}\Bigg(\Big[D(I-A^\#+A_d^\#E)\Big]_{1j}\rho\Bigg),$$
		where if $j=1$, we assume $\rho$ is any pure state in $V$  (mean recurrence time to $V$ given $\rho$), or if $j=2$, then $\rho$ is any pure state in $V^\perp$.
	\end{theorem}
	
	In a similar way as for Theorem \ref{mhtf}, this result is also valid for the case that the initial state and the goal subspace are not orthogonal, see \cite{cv}.
	
	\medskip
	
	The main conclusion that can be drawn from the above theorem goes in line with previous results on irreducible QMCs, namely, that the mean hitting time from a state to a goal subspace is essentially determined by the mean return times given by the diagonal of $K$, with some generalized inverse performing an auxilliary role in obtaining the analytic expression.
	
	\bigskip
	
{\bf Justifying the use of Assumption I.} Let us conclude this section by explaining how Assumption I has appeared in our study. First, we recall the classical expression on the hitting time $k_{ij}$ from vertex $j$ to vertex $i$ \cite{kemeny},
	$$k_{ij}=1+\sum_{l\neq i}p_{lj}k_{il},$$
	the so-called conditioning of the first step: the mean time equals 1 plus the weighted average of the probabilities $p_{lj}$ of moving from $j$ to some other state $l\neq i$ multiplied by the mean time of moving from $l$ to $i$, that is, $k_{il}$. If $P=(p_{ij})$, $K=(k_{ij})$, $D=diag(K)$, and $E$ the matrix whose entries are all equal to 1, this can be written in matrix form as
$$K=E+(K-D)P \;\Rightarrow\; K-(K-D)P=E.$$
Now turning to our setting, if $\Lambda$ is a QMC, one has a quantum version of the above expression, namely, one defines the operator
$$L=K-(K-D)\Lambda,$$
where $K$ and $D$ are the corresponding operators obtained from $\Lambda$, and it is known that, for every choice of density $\rho$ concentrated at $j$ we have $\mathrm{Tr}(L_{ij}\rho)=1$, where $L_{ij}$ is the operator (block matrix) at position $(i,j)$, see \cite{c2017,cv}. 

\medskip

Now, note that operator $L$ makes sense whenever each of the hitting time operators $K_{ij}$ produce finite hitting times for every choice of density (actually one only needs $K_{11}$ and $K_{12}$ to be finite). How to ensure such finiteness? Here we have the crucial step: in \cite{glv}, it is proven that if 1 does not belong to the spectrum of $\mathbb{Q}\Lambda$, we have the analyticity of the generating function $\mathbb{G}(z)=z(I-z\mathbb{Q}\Lambda)^{-1}$ at $z=1$, which implies the finiteness of the mean hitting times for any initial state, besides ensuring that the probability of ever reaching the goal subspace equals 1, also see [\cite{cv}, Remark 3.1]. This provides a global condition that allows us to write operator $L$ and to obtain hitting time formulae in terms of a kernel which is well-defined for the entire space. The fact that we are able to write a well-defined operator $L$ is at the core of the proof of Theorem \ref{BeyondIrreducibleTheorem}. As it becomes clear in the next section, there are reducible examples which also satisfy this spectral condition, thus extending previous results seen in the literature.

\section{Examples}

\medskip

\subsection{Quantum channel on 1 qubit} The following example was studied in \cite{cv}, and here we review it under the results of this article.

\bex
	 Consider the quantum channel $T$ acting on $M_2(\mathbb C)$ given by $T(X) = AXA^*+BXB^*$, where 
	\begin{align*}
		A=\frac{1}{\sqrt{3}}\begin{bmatrix}
			1 & 1\\ 0 & 1\end{bmatrix}, \quad B=\frac{1}{\sqrt{3}}\begin{bmatrix}
			1 & 0\\ -1 & 1\end{bmatrix}.
	\end{align*}
	We have that $T$ is trace preserving and unital (i.e., it is identity-preserving). The matrix representation of the channel is given by 
	\begin{align*}
		\ceil T =A\otimes \Bar{A}+B\otimes\Bar{B}= \frac{1}{3}\begin{bmatrix}
			2 & 1 & 1 & 1 \\
			-1 & 2 & 0 & 1 \\
			-1 & 0 & 2 & 1 \\
			1 & -1 & -1 & 2
		\end{bmatrix}.
	\end{align*}
We choose two  orthogonal states,
	\begin{align*}
		\phi = \frac{1}{\sqrt{2}}\begin{bmatrix}
			1 \\ -1 \end{bmatrix}, \quad \psi = \frac{1}{\sqrt{2}}\begin{bmatrix}
			1 \\1\end{bmatrix},
	\end{align*} and define their projections,
	\begin{align*}
		P=|\psi\rangle \langle \psi| = \frac{1}{2}\begin{bmatrix}
			1 & 1 \\ 1 & 1\end{bmatrix}, \quad Q = |\phi\rangle \langle \phi| = \frac{1}{2}\begin{bmatrix}
			1 & -1 \\ -1 & 1\end{bmatrix}.
	\end{align*}
	With these, the orthogonal projector maps acting on $M_2$ given by $\mathbb P = P\cdot P$ and $\mathbb Q = Q\cdot Q$ have matrix representations
	\begin{align*}
		\ceil {\mathbb P} =P\otimes \Bar{P}= \frac{1}{4}\begin{bmatrix}
			1 & 1 & 1 & 1 \\
			1 & 1 & 1 & 1 \\
			1 & 1 & 1 & 1 \\
			1 & 1 & 1 & 1
		\end{bmatrix}, \quad \ceil {\mathbb Q} =  Q\otimes \Bar{Q}=\frac{1}{4} \begin{bmatrix}
			1 & -1 & -1 & 1 \\
			-1 & 1 & 1 & -1 \\
			-1 & 1 & 1 & -1 \\
			1 & -1 & -1 & 1
		\end{bmatrix}.
	\end{align*}
	\
	
	The matrix representation of the mean hitting time operator for $T$ is given by
	\begin{align*}
		\ceil K = \ceil T(I_4-\ceil{\mathbb Q}\ceil T)^{-2}=\frac{1}{6}\begin{bmatrix}
			39 & -12 & -12 & 9 \\
			-72 & 32 & 28 & -12 \\
			-72 & 28 & 32 & -12 \\
			177 & -72 & -72 & 39
		\end{bmatrix}.
	\end{align*}
	
	With this, we can calculate $\tau(\phi\rightarrow V)$, where $V=\ger \{\psi\}\subset \mathbb C^2$, and $\rho_\phi =|\phi\rangle \langle \phi|=Q$. We have 
	\begin{align}\label{exampleQC}
		\tau(\phi\rightarrow V) = \Tr(\mathbb PK\rho_\phi)= 6.
	\end{align}
	
	Now, alternatively, one may turn to the KSMH formula for hitting times in order to verify the result above. First, we define the QMC $\Lambda$ via equation (\ref{LambdaChannel}), whose matrix representation is 
	\begin{align*}
		\ceil \Lambda = \begin{bmatrix}
			\ceil{\mathbb P}\ceil T & \ceil{\mathbb P}\ceil T \\
			\ceil{\mathbb Q}\ceil T & \ceil{\mathbb Q}\ceil T
		\end{bmatrix} =\frac{1}{12} \begin{bmatrix}
		1 & 2 & 2 & 5 & 1 & 2 & 2 & 5\\
		1 & 2 & 2 & 5 & 1 & 2 & 2 & 5\\
		1 & 2 & 2 & 5 & 1 & 2 & 2 & 5\\
		1 & 2 & 2 & 5 & 1 & 2 & 2 & 5\\
		5 & -2 & -2 & 1 & 5 & -2 & -2 & 1\\
		-5 & 2 & 2 & -1 & -5 & 2 & 2 & -1\\
		-5 & 2 & 2 & -1 & -5 & 2 & 2 & -1\\
		5 & -2 & -2 & 1 & 5 & -2 & -2 & 1
		\end{bmatrix}.		
	\end{align*}	
In our case, matrix $E$ equals
\beq\label{e224}
		E = \begin{bmatrix} I_4 & I_4 \\ I_4 & I_4\end{bmatrix}.
\eeq
	Then, we have $A=I-\Lambda$,
	$$A^\#=\frac{1}{4}\begin{bmatrix} 
	1 & 2 & 2 & 1 & -3 & 2 & 2 & 1\\
	-3 & 6 & 2 & 1 & -3 & 2 & 2 & 1\\
	-3 & 2 & 6 & 1 & -3 & 2 & 2 & 1 \\
	-3 & 2 & 2 & 5 & -3 & 2 & 2 & 1 \\
	1 & -2 & -2 & -3 & 5 & -2 & -2 & -3 \\
	-1 & 2 & 2 & 3 & -1 & 6 &2 & 3\\
	-1 & 2 & 2 & 3 & -1 & 2 & 6 & 3\\
	1 & -2 & -2 & -3 & 1 & -2 & -2 & 1
	\end{bmatrix},$$
	so that
	\begin{align*}
		\tau(\phi \rightarrow V) &= \Tr\big(\big[D(I-A^\#+A_d^\#E)\big]_{12}\rho_\phi\big) \\
		&= \Tr \left(
		 \begin{bmatrix} 3 & -1 & -1 & 1\\
			3 & -1 & -1 & 1\\
			3 & -1 & -1 & 1\\
			3 & -1 & -1 & 1
		\end{bmatrix}\cdot \begin{bmatrix} 1/2 \\ -1/2 \\ -1/2 \\ 1/2 \end{bmatrix} \right)  = \Tr\left( \begin{bmatrix}
			3 \\ 3 \\ 3\\ 3
		\end{bmatrix}\right) = 3+3 = 6,
	\end{align*}
	which is the same result we obtained above in (\ref{exampleQC}), as expected.

\eex\qee

\medskip

\subsection{Unitary walks}\label{sec8}

In order to study hitting times for unitary dynamics in a consistent manner, first we need to ensure that the probability of ever reaching a goal subspace equals 1 for some initial state. Given a unitary matrix $U$, we consider the quantum channel induced by it, namely the conjugation map $\mathbb{U}:M_n\to M_n$, $\mathbb{U}(X)=UXU^*$. Then, as its representation matrix is $U\otimes\ov{U}$, the value 1 will be a degenerate eigenvalue for $\mathbb{U}$, so the conjugation map induced by a unitary matrix will be reducible. Nevertheless, as we have seen in this work, a proper assumption on $\mathbb{Q}\mathbb{U}$, where $\mathbb{Q}$ is the projection onto the complement of the target space, allows us to write hitting time formulae in terms of the associated group inverse of the walk.

\bex (Order 2 unitary matrix) Let $U\in M_2$ be the Hadamard matrix and $\mathbb{U}$ the associated conjugation,
$$U=\frac{1}{\sqrt{2}}\begin{bmatrix} 1 & 1 \\ 1 & -1\end{bmatrix},\;\;\;\mathbb{U}(X)=UXU^*\;\Rightarrow\; \lceil\mathbb{U}\rceil=U\otimes\ov{U}=\frac{1}{2}\begin{bmatrix} 1 & 1 & 1 & 1 \\ 1 & -1 & 1 & -1 \\ 1 & 1 & -1 & -1 \\ 1 & -1 & -1 & 1  \end{bmatrix}.$$
It holds that $U$ has a simple spectrum, but $\mathbb{U}$ has 1 as an eigenvalue with multiplicity 2 (two strictly positive eigenstates are easily obtained, one of which is a multiple of the identity). As a consequence, $\mathbb{U}$ is reducible. Nevertheless, there are many subspaces for which one may calculate hitting times in terms of g-inverses. For instance, if we consider any state of the form 
$$\psi=\begin{bmatrix}\alpha \\ \sqrt{1-\alpha^2}\end{bmatrix},\;\;\;0\leq \alpha\leq 1,\;\;\;\alpha\neq \frac{1}{2}\sqrt{2+\sqrt{2}}, \;\;\; P=|\psi\rangle\langle\psi|,\;\;\;Q=I-P,\;\;\;\mathbb{P}=P\cdot P,\;\;\;\mathbb{Q}=Q\cdot Q,$$
then $1$ is not in the spectrum of $\mathbb{Q}\mathbb{U}$. For a particular case, take $\alpha=1$ so that $\psi=[1 \; 0]^T$, and let $\phi=[0\; 1]^T$. A straighforward calculation gives that
$$K=\ceil{\mathbb{U}}(I-\ceil{\mathbb{Q}}\ceil{\mathbb{U}})^{-2}=\begin{bmatrix} 2 & -1 & -1 & 2 \\-1 & 1 & 2 & -2 \\ -1 & 2 & 1 & -2 \\2 & -2 & -2 & 2\end{bmatrix},$$
so
$$\tau(\phi\to \psi)=\mathrm{Tr}(\ceil{\mathbb{P}}\ceil{\mathbb{U}}(I-\ceil{\mathbb{Q}}\ceil{\mathbb{U}})^{-2}(\rho_\phi))=2.$$
In terms of Theorem \ref{BeyondIrreducibleTheorem}, the QMC induced by $\mathbb{U}$ and $V=span\{\psi\}$ is
$$\Lambda=\frac{1}{2}\begin{bmatrix}
1 & 1 & 1 &1 & 1 & 1 & 1 & 1\\
0 & 0 & 0 & 0 & 0 & 0 & 0 & 0 \\
0 & 0 & 0 & 0 & 0 & 0 & 0 & 0 \\
0 & 0 & 0 & 0 & 0 & 0 & 0 & 0 \\
0 & 0 & 0 & 0 & 0 & 0 & 0 & 0 \\
0 & 0 & 0 & 0 & 0 & 0 & 0 & 0 \\
0 & 0 & 0 & 0 & 0 & 0 & 0 & 0 \\
1 & -1 & -1 &1 & 1 & -1 & -1 & 1
\end{bmatrix},$$
the group inverse associated with $\Lambda$ is 
$$A=I-\Lambda,\;\;\;A^\#=\frac{1}{2}\begin{bmatrix} 1 & 1 & 1 & -1 & -1 & 1 & 1 & -1\\
0 & 2 & 0 & 0 & 0 & 0 & 0 & 0 \\
0 & 0 & 2 & 0 & 0 & 0 & 0 & 0 \\
0 & 0 & 0 & 2 & 0 & 0 & 0 & 0 \\
0 & 0 & 0 & 0 & 2 & 0 & 0 & 0 \\
0 & 0 & 0 & 0 & 0 & 2 & 0 & 0 \\
0 & 0 & 0 & 0 & 0 & 0 & 2 & 0 \\
-1 & -1 & -1 & -1 & -1 & -1 & -1 & 1
\end{bmatrix},$$
so the KSMH kernel can be obtained and provides, as expected,
$$D[I-A^\#+A_d^\#E]_{12}=\begin{bmatrix} 
 2 & -1 & -1 & 2\\
 0 & 0 & 0 & 0 \\
 0 & 0 & 0 & 0 \\
 0 & 0 & 0 & 0 \end{bmatrix}\;\Rightarrow\; \tau(\phi\to \psi)=\mathrm{Tr}\Bigg(\Big[D(I-A^\#+A_d^\#E)\Big]_{12}\rho_\phi\Bigg)=2.$$	

\eex
\qee

\bex (Order 4 unitary matrix) Let
$$U=\frac{1}{\sqrt{2}}\begin{bmatrix} 1 & 1 & 0 & 0 \\ 0 & 0 & 1 & 1 \\ 1 & -1 & 0 & 0 \\ 0 & 0 & 1 & -1\end{bmatrix},\;\;\;\mathbb{U}(X)=UXU^*,$$
which corresponds to a coined walk on two vertices, and let
$$|\psi\rangle=\begin{bmatrix} 1 \\ 0 \\0 \\0\end{bmatrix},\;\;\; |\phi\rangle=\begin{bmatrix}0\\ \alpha\\ \beta \\ \delta\end{bmatrix},\;\;\;|\alpha|^2+|\beta|^2+|\delta|^2=1,\;\;\;P=|\psi\rangle\langle\psi|,\;\;\;Q=I-P,\;\;\;\mathbb{P}=P\cdot P, \;\;\;\mathbb{Q}=Q\cdot Q.$$
Although relatively large, a computer-aided computation produces the results  in a straightforward manner, both in terms of the definition of hitting times acting on the order 16 representation matrix of $\mathbb{U}$ and also in terms of the induced QMC (order 32 matrix) with Theorem \ref{BeyondIrreducibleTheorem}. In this case the group inverse is easily obtained as well (the individual entries are all simple fractions, but we omit the explicit matrices for brevity). As 1 is not in the spectrum of $\mathbb{Q}\mathbb{U}$, we can write the hitting time operator as $K=\mathbb{U}(I-\mathbb{Q}\mathbb{U})^{-2}$. Explicitly, we have
$$K=\begin{bmatrix} B_1 & B_2 \\ B_3 & B_4\end{bmatrix},$$
where, for simplicity, $K$ is written as the 4 order 8 matrices $B_i$ given by
$$B_1=\begin{bmatrix} 
4 & -3 & -1 & 2 & -3 & 4 & 1  & -2\\
-1 & 1 & -1 & -5 & 1 & -1 & 2 & 6 \\
-3 & 3 & 1 & -2 & 4 & -4 & -1 & 2 \\
2 & -2 & -5 & 8 & -2 & 2 & 6 & -9 \\
-1 & 1 & 6 & -3 & 1 & -1 & -6 & 3\\
6 & -6 & -2 & 4 & -6 &  6 & 2 & -4 \\
1 & -1 & -6 & 3 & -1 & 1 & 6 & -3 \\
-3 & 3 & 4 & -13 & 3 & -3 & -4 & 13
\end{bmatrix},\;\;\;B_2=\begin{bmatrix}
-1 & 1 & 6 & -3 & 2 & -2 & -3 & 10 \\
6 & -6 & -2 & 4 & -3 & 3 & 8 & -12\\
1 & -1 & -6 & 3 & -2 & 2 & 3 & -10\\
-3 & 3 & 4 & -13 & 10 & -10 & -12 & 25\\
-1 & 2 & -2 & 8 & -5 & 6 & 4 & -12\\
-2 & 2 & 10 & -6 & 4 & -4 & -6 & 18\\
2 & -2 & 2 & -8 & 6 & -6 & -4 & 12\\
8 & -8 & -6 & 16 & -12 & 12 & 18 & -34
\end{bmatrix},$$
$$B_3=\begin{bmatrix}
-3 & 4 & 1 &-2 & 3 & -4 & -1 & 2\\
1 & -1 & 2 & 6 & -1 & 1 & -2 & -6\\
4 & -4 & -1 & 2 & -4 & 4 & 1 & -2\\
-2 & 2 & 6 & -9 & 2 & -2 & -6 & 9\\
2 & -2 & -3 & 10 & -2 & 2 & 3 & -10\\
-3 & 3 & 8 & -12 & 3 & -3 & -8 & 12\\
-2 & 2 & 3 & -10 & 2 & -2 & -3 & 10 \\
10 & -10 & -12 & 25 & -10 & 10 & 12 & -25
\end{bmatrix},\;\;\;B_4=\begin{bmatrix}
1 & -1 & -6 & 3 & -2 & 2 & 3 & -10\\
-6 & 6 & 2 & -4 & 3 & -3 & -8 & 12\\
-1 & 1 & 6 & -3 & 2 & -2 & -3 & 10\\
3 & -3 & -4 & 13 & -10 & 10 & 12 & -25\\
-5 & 6 & 4 & -12 & 8 & -9 & -13 & 25\\
4 & -4 & -6 & 18 & -13 & 13 & 16 & -34\\
6 & -6 & -4 & 12 & -9 & 9 & 13 & -25\\
-12 & 12 & 18 & -34 & 25 & -25 & -34 & 72
\end{bmatrix},$$
from which we obtain 
$$\tau(\phi\to\psi)=\mathrm{Tr}(\mathbb{P}K\rho_\phi)=4|\alpha|^2+6|\beta|^2+10|\delta|^2+2Re(\alpha\ov{\beta})-4Re(\alpha\ov{\delta})-6Re(\beta\ov{\delta}).$$
The particular case $\alpha=1$, $\beta=\delta=0$ (i.e., $\phi$ equals the canonical vector $e_2$) gives $\tau(e_2\to\psi)=4$ and, similarly, $\tau(e_3\to\psi)=6$ and $\tau(e_4\to\psi)=10$. Proper calculations lead us to apply Theorem \ref{BeyondIrreducibleTheorem}, producing the same result as above, as expected. 
\eex

\qee

\begin{remark} The reader should compare the notion of hitting times presented in this work with the one seen in \cite{portugal,szegedy}, associated with Szegedy's walk, for which such quantity is not defined in terms of monitoring of subspaces. A difference worth observing is that in Szegedy's walk, one usually takes the maximally mixed state as the initial state, which in general allows us to circumvent the problem of the degeneracy of the spectrum of the unitary operator. The discussion regarding whether such hitting time notions are equivalent in some sense remains, up to our knowledge, an open question.\end{remark}

\section{Summary}\label{sec9}

In this work, we have discussed the problem of obtaining mean times of first visit to some chosen goal subspace, given some initial state, under the dynamics of a quantum channel acting on a finite-dimensional Hilbert space. We have extended previous results on this matter by considering a condition on the channel which is strictly weaker than the assumption of irreducibility. As a consequence, now we are able to consider hitting times of unitary conjugations in terms of generalized inverses. The procedure can be briefly summarized as follows: let $T$ be a quantum channel, $V$ a goal subspace with associated orthogonal projection $P$, $Q=I-P$, $\mathbb{Q}=Q\cdot Q$ and let $\Lambda=\Lambda_{T,V}$ be the induced QMC. Then, if 1 does not belong to the spectrum of $\mathbb{Q}\Lambda$, one has a global condition that allows us to obtain hitting times to subspace $V$, with respect to the dynamics given by $T$, in terms of a generalized inverse of $I-\Lambda$, namely the group inverse, and an operator $D_{11}$ characterizing the mean return times for $V$. The explicit formula is given by Theorem \ref{BeyondIrreducibleTheorem}.

\medskip

A future research direction concerns the problem of generalizing the present results to the setting of quantum maps acting on infinite-dimensional vector spaces. In principle, one has to study generalized inverses while taking in consideration positive maps in such context. Regarding the case of unitary quantum walks on the line, we believe that a natural starting point consists of considering the CMV decomposition of unitary maps \cite{cmv}, which allows us to write the evolution as a direct sum of unitaries with a simple spectrum. A discussion of these matters will be the topic of a future work.

\bigskip

{\bf Acknowledgments.} The contents of this work have appeared, or are a refinement of results from the PhD thesis of one of the authors (LFLP), see \cite{lflp_tese}. The authors are grateful to the referees for several suggestions regarding the manuscript. CFL is would like to thank the organizers of the Graph Theory, Algebraic Combinatorics and Mathematical Physics Conference in Montr\'eal, Canada, during which a lecture on group inverses and quantum walks was presented, and to L. Vel\'azquez and R. Portugal, regarding discussions on quantum probabilities and hitting times in the unitary setting. LFLP acknowledges financial support from CAPES (Coordena\c c\~ao de Aperfei\c coamento de Pessoal de N\'ivel Superior - Brasil) during the period 2018-2021. This publication is part of the I+D+i project PID2021-124472NB-I00 funded by MCIN/AEI/10.13039/501100011033/ and ERDF “Una manera de hacer Europa”.

\medskip

\section{Appendix}

\subsection{Proof of Lemma \ref{mnov_lemma}.} The proof is inspired by [\cite{meyer-artigo}, Thm. 2.2], which is valid for stochastic matrices. For the setting of positive maps, we proceed as follows: first, use Theorem \ref{wolfproposition} in order to write
\begin{align*}
		\ceil T = X\begin{bmatrix}
			I &  \\  & B
		\end{bmatrix}X^{-1},
	\end{align*} where $I$ is an identity matrix of some dimension and $B$ is a matrix with no eigenvalue equal to 1. Then, write
	$$\frac{I+B+B^2+\cdots+B^{n-1}}{n}=\frac{(I-B^n)(I-B)^{-1}}{n}.$$
	From this, and with the factorizations
	$$A=X\begin{bmatrix} O & \\ & I-B\end{bmatrix}X^{-1},\;\;\;A^\#=X\begin{bmatrix} O & \\ & (I-B)^{-1}\end{bmatrix}X^{-1},$$	
	we obtain that
	$$\frac{I+T+\cdots+T^{n-1}}{n}=\frac{(I-T^n)A^\#}{n}+I-AA^\#.$$
	By Russo-Dye's theorem \cite{bhatia}, the positivity and trace preservation of $T$ implies that $\Vert T^n\Vert=1$ for all $n$. Therefore,
	$$\lim_{n\to\infty} \frac{(I-T^n)A^\#}{n}=0,$$
	and hence
$$\lim_{n\to\infty}\frac{I+T+\cdots+T^{n-1}}{n}=I-AA^\#.$$
Now, note that $(I-T)(I-A^\#A)=A(I-A^\#A)=A-AA^\#A=A-A=0$, which implies that $(I-A^\#A)\rho$ is invariant for $T$ for any choice of matrix $\rho$. Finally, if $\rho$ is a density matrix, so is $(I-A^\#A)\rho$. In fact, consider the sequence $T_n=(I+T+\cdots+T^{n-1})/n$, $n=1,2,\dots$, which approximates $(I-A^\#A)$ uniformly. We have that $T_n\rho$ is a density matrix for all $n$ and, as the set of density matrices in $M_n$ is compact, the result follows.

\qed

\subsection{Proof of Theorem \ref{BeyondIrreducibleTheorem}.} First, we recall a technical result described by Hunter \cite{hunter}:

	\begin{theorem}\cite{hunter}\label{theoremofginverses}
		A necessary and sufficient condition for the equation $FXB=C$ to have a solution is that $FF^-CB^-B=C$, where $F^-$ and $B^-$ are any g-inverses for $F$ and $B$, respectively. In this case, the general solution is given by one of the two equivalent forms: either
		\begin{align*}
			X=F^-CB^-+H-F^-DHBB^-,
		\end{align*} where $H$ is an arbitrary matrix, or
				\begin{align*}
			X=F^-CB^-+(I-F^-F)U+V(I-B^-B),
		\end{align*}
		where $U$ and $V$ are arbirtrary matrices.		
	\end{theorem}

\medskip

Now we recall a basic lemma, regarding a conditioning on the first step reasoning. The probabilitistic reasoning given in terms of the operator $L$ has been discussed in \cite{c2017} in the context of OQWs, but the same proof holds for the case of QMCs, also see \cite{cv} for the case of positive maps.
	
	\begin{lemma}\label{propertyL}(Adapted from [\cite{c2017}, Lemma 2]).	Let $\Phi$ be a QMC, let $K$ be its hitting time operator and $D=\diag(K_{11},\dots,K_{nn})$, and define $L:=K-(K-D)\Phi$. Then for $\rho_j$ a density concentrated at site $j$, for all $i$, it holds that $\Tr\big(L_{ij}(\rho_j)\big)=\Tr(\rho_j)=1$. As a consequence, $\Tr\big(L_{ij}\rho)=Tr(\rho)$ for every $\rho\in M_n$.
	\end{lemma}

Together with the results stated above, the following proof is inspired by ideas seen in \cite{oqwmhtf2,meyer-artigo}.

\medskip
	
{\bf Proof of Theorem \ref{BeyondIrreducibleTheorem}.} We start with the definition $L:=K-(K-D)\Lambda$, where $D=K_d$, and rearrange it to obtain 
	\begin{align}\label{equationforKgroupinv}
		K(I-\Lambda)=L-D\Lambda.
	\end{align} 
	Operator $L$ is well-defined due to Assumption I (also see the remark at the end of Section \ref{sec7}). Now define $A:=I-\Lambda$ and we consider its group inverse $A^\#$, which exists by Proposition \ref{groupinverseexists}.  Now, in Theorem \ref{theoremofginverses}, take $F=I$, $X=K$, $B=A$ and $C=L-D\Lambda$, so we have that equation (\ref{equationforKgroupinv}) has a solution $K$ if, and only if, $F^-FCB^-B=C$, which is equivalent to $(L-D\Lambda)A^\#A=L-D\Lambda$, that is,
	\begin{align*}
		(L-D\Lambda)(I-A^\#A)=0 \quad \Longleftrightarrow \quad KA(I-A^\#A)=0.
	\end{align*} But it is clear that this last equation is satisfied, due to the property $AA^\#A=A$. So, we have that the solution to (\ref{equationforKgroupinv}) can be written as 
	\begin{align}\label{solutionKgroupinv}
		K=(L-D\Lambda)A^\#+V(I-AA^\#), 
	\end{align}where $V$ is an arbitrary matrix.
	
	\medskip
	
By Lemma \ref{mnov_lemma}, the operator $I-A^\#A$ projects onto the space of fixed points of $\Lambda$. Moreover, we note that for each $i$, vector $(I-A^\#A)|e_i\rangle$ is the $i$-th column of $I-A^\#A$, and that these are fixed points for $\Lambda$, due to the property $AA^\#A=A$. Also, by the fact that $\Lambda$ is a QMC acting on two vertices, we may write
	 \begin{align}\label{eqcols}
	 I-A^\#A=\begin{bmatrix}
			X & X\\ Y & Y
		\end{bmatrix},
	\end{align}
	 for certain $X, Y\in M_{k^2}$. The reason for this is simple: note that for any choice of $\rho_1$, $\rho_2$,
	 \begin{align}
	 \begin{bmatrix}
			X & X\\ Y & Y
		\end{bmatrix}\begin{bmatrix} \rho_1 \\ \rho_2\end{bmatrix}=\begin{bmatrix} X(\rho_1+\rho_2) \\ Y(\rho_1+\rho_2)\end{bmatrix},
			\end{align}
		so that any choice of density will be projected onto the image of $X$ and $Y$, with respect to vertices 1 and 2, respectively. It is then clear that the analogous behavior will apply to $V(I-AA^\#)$, for any matrix $V$, that is, these will be matrices of the form (\ref{eqcols}) accordingly. Due to these facts, if we define 
	\begin{align*}
		E:=\begin{bmatrix}
			I_{k^2} & I_{k^2} \\ I_{k^2}  & I_{k^2}
		\end{bmatrix},
	\end{align*} we will have $(I-AA^\#)_d E = (I-AA^\#)$ and $\big(V(I-AA^\#)\big)_d E = V(I-AA^\#)$. 
	
	\medskip
	
	Define $B:=\big(V(I-AA^\#)\big)_d$ and substitute $V(I-AA^\#)=BE$ in (\ref{solutionKgroupinv}) in order to write
	\begin{align}\label{equationforK2groupinv}
		K=(L-D\Lambda)A^\#+BE.
	\end{align} We take the diagonal and obtain \begin{align*}
		D=K_d =(LA^\#)_d -D(\Lambda A^\#)_d + B
		\quad \Longrightarrow \quad B=D+D(\Lambda A^\#)_d-(LA^\#)_d.
	\end{align*}
	Now we define $W:=L-D(I-AA^\#)$ and use it to eliminate $L$ in the equation for $B$ above, so we get:
	\begin{align}\label{equationforBgroupinv}
		B&=D+D(\Lambda A^\#)_d - (WA^\#)_d- \big(D(I-AA^\#)A^\#\big)_d \nonumber \\
		&=D+D(\Lambda A^\#)_d - (WA^\#)_d,
	\end{align} where one term was canceled since $(I-AA^\#)A^\#=(I-A^\#A)A^\#=A^\#-A^\#AA^\#=0$. Substituting $B$ from (\ref{equationforBgroupinv}) in (\ref{equationforK2groupinv}), and eliminating $L$ using $W$, we have
	\begin{align}\label{equationforK3groupinv}
		K&=\Big(W+D(I-AA^\#)-D\Lambda\Big)A^\#+DE+D(\Lambda A^\#)_dE -(WA^\#)_dE \nonumber \\
		&=D\Big[-\Lambda A^\#+E+(\Lambda A^\#)_dE\Big] +WA^\#-(WA^\#)_dE,
	\end{align} where again a term has vanished since $(I-AA^\#)A^\#=0$. 	Now consider $H:=(I-AA^\#)_d$. We know $HE=I-AA^\#$, so 
	\begin{align*}
		I-AA^\#=I-A^\#+\Lambda A^\# = HE,
	\end{align*} and hence, taking the diagonal we obtain 
	\begin{align}\label{simplificationgroupinv}
		& I-(A^\#)_d+(\Lambda A^\#)_d=H \nonumber\\
		\Longrightarrow &E-(A^\#)_dE+(\Lambda A^\#)_dE= HE = I-A^\#+\Lambda A^\# \nonumber \\
		\Longrightarrow & -\Lambda A^\#+E+(\Lambda A^\#)_dE = I-A^\#+A^\#_dE.
	\end{align} By replacing (\ref{simplificationgroupinv}) into (\ref{equationforK3groupinv}), we obtain
	\begin{align}\label{solutionfinalforKgroupinv}
		K=D\Big[I-A^\#+A^\#_dE\Big]+WA^\#-(WA^\#)_dE.
	\end{align} It remains to show that $\Tr\Big((WA^\#)_{12}\rho\Big)$ and $\Tr\Big(\big[(WA^\#)_dE\big]_{12}\rho\Big)$ are zero for any $\rho$.  Note that by multiplying (\ref{equationforKgroupinv}) on the right by any $\Lambda$-invariant vector $|\rho\rangle$, we obtain
	$ L|\rho\rangle = D|\rho\rangle$, whence
	\begin{align}\label{kacgroupinv}
		\Tr\big(K_{11}|\rho_1\rangle\big) = \Tr\big((L|\rho\rangle)_1\big)=\sum_j \Tr\big(L_{1j}|\rho_j\rangle\big) 
		= \sum_j \Tr(\rho_j),
	\end{align} where in the third equality we have used Lemma \ref{propertyL}.
	Also recall that for any vector $|\rho\rangle$, the new vector $(I-AA^\#)|\rho\rangle$ will be a state which is invariant by $\Lambda$, with 
	\begin{align}\label{trace1groupinv}
		\sum_j \Tr\Big(\big((I-AA^\#)|\rho\rangle\big)_j\Big)&= \Tr\Big((I-AA^\#)|\rho\rangle \Big) \nonumber \\ &
		=\langle e_{I^2_{k^2}}|(I-AA^\#)|\rho\rangle =\langle e_{I^2_{k^2}}|\rho\rangle \nonumber \\
		&=\sum_j \langle e_{I_{k^2}}|\rho_j\rangle = \sum_j \Tr(\rho_j).
	\end{align} 
	If $|\rho\rangle$ is a vector concentrated at site $m$, i.e., if it is of the form
	\begin{align*}
		|\rho\rangle = \begin{bmatrix}
			\rho \\ 0
		\end{bmatrix} \text{ or } \begin{bmatrix}
			0 \\ \rho
		\end{bmatrix},
	\end{align*} then equation (\ref{trace1groupinv}) reduces to 
	\begin{align}\label{traceinvgroupinv}
		\sum_j \Tr\Big(\big((I-AA^\#)|\rho\rangle\big)_j\Big) = \Tr\big(\rho\big).
	\end{align}\
	
	Now we proceed to show that the terms involving $W$ in (\ref{solutionfinalforKgroupinv}) will have trace zero. First, we have that
	\begin{align}\label{firstWtracegroupinv}
		\Tr\Big((WA^\#)_{1r}\rho\Big) &= \Tr\big((LA^\#)_{1r}\rho\big)-\Tr\Big(\big(D(I-AA^\#)A^\#\big)_{1r}\rho\Big) \nonumber \\
		&=\sum_m \left[ \Tr\Big(L_{1m}A^\#_{mr}\rho\Big)-\Tr\Big(K_{11}(I-AA^\#)_{1m}A^\#_{mr}\rho\Big)\right]\nonumber \\
		&=\sum_m \left[ \Tr\Big(A^\#_{mr}\rho\Big)-\Tr\Big(K_{11}(I-AA^\#)_{1m}A^\#_{mr}\rho\Big)\right],
	\end{align} where again Lemma \ref{propertyL} was applied for $L$, and index $r$ can be either 1 or 2. Note that 
	\begin{align*}
		(I-AA^\#)_{1m}A^\#_{mr}\rho = \big[(I-AA^\#)|\rho \rangle\big]_1,
	\end{align*} where $|\rho\rangle$ is the vector with $A^\#_{mr}\rho$ concentrated at site $m$. So, for this choice of $|\rho\rangle$, 
	\begin{align}\label{firstWconclusiongroupinv}
		\Tr\Big(K_{11}(I-AA^\#)_{1m}A^\#_{mr}\rho\Big)&=\Tr\Big(K_{11}\big[(I-AA^\#)|\rho \rangle\big]_1\Big) \nonumber \\
		&=\sum_j\Tr\Big(\big[(I-AA^\#)|\rho\rangle\big]_j\Big)\nonumber \\
		&=\Tr\Big(A^\#_{mr}\rho\Big),
	\end{align} where in the second equality we used equation (\ref{kacgroupinv}), and in the last equality we used (\ref{traceinvgroupinv}). Inserting (\ref{firstWconclusiongroupinv}) back into (\ref{firstWtracegroupinv}) we get
	\begin{align*}
		\Tr\Big((WA^\#)_{1r}\rho\Big) = \sum_m \left[\Tr\Big(A^\#_{mr}\rho\Big)-\Tr\Big(A^\#_{mr}\rho\Big)\right]=0.
	\end{align*} It is immediate from the above with $r=2$ that $\Tr\Big((WA^\#)_{12}\rho\Big)=0$. But also for $r=1$ it gives us
	\begin{align*}
		\Tr\Big(\big((WA^\#)_dE\big)_{12}\rho\Big)=\Tr\Big((WA^\#)_{11}\rho\Big)=0.
	\end{align*}\
	
	Therefore, when we calculate $\Tr\big(K_{12}\rho\big)$ using (\ref{solutionfinalforKgroupinv}), the terms involving $W$ vanish and the result follows.
	
\qed

{\bf Conflict of interest statement.} The authors declare that they have no conflict of interest.

\medskip

{\bf Data availability statement.} The datasets generated during the current study are available from the corresponding author on reasonable request.

\end{document}